\documentclass[aps,pre,twocolumn,showpacs,10pt]{revtex4-2}
\usepackage{graphicx}
\usepackage{amsmath,amssymb}
\usepackage{amsfonts}
\usepackage[dvipsnames]{xcolor}
\usepackage{enumitem}
\usepackage{amsthm}
\newtheorem{defn}{Definition} \newtheorem{prop}[defn]{Proposition} 
\newtheorem{cor}{Corollary}
\begin{document}

\title{Incidence-based random walks on simplicial complexes} 

\author{C.T. Martínez-Martínez}
\email[]{claudiam@fisica.unam.mx}
\affiliation{Instituto de F\'isica, Universidad Nacional Aut\'onoma de M\'exico, Ciudad de M\'exico, C.P. 04510, Mexico}

\author{Francisco J. Sevilla}
\email[]{fjsevilla@fisica.unam.mx}
\affiliation{Instituto de F\'isica, Universidad Nacional Aut\'onoma de M\'exico, Ciudad de M\'exico, C.P. 04510, Mexico}

\date{\today}

\begin{abstract}
We introduce an incidence-based random walk on the edges of a random two-dimensional simplicial complex with a complete $1$-skeleton and independently retained triangular faces. The dynamics combine two transport channels, one mediated by vertices and the other by triangular faces, through an effective transition operator controlled by a mixing parameter $q$. This construction isolates the effects of higher-order connectivity without modifying the underlying pairwise support of the walk. We characterize the model through structural observables, spectral relaxation, stationary localization, and first-passage transport. Our results show that partial face retention generates heterogeneous higher-order connectivity, giving rise to a pronounced transport bottleneck at intermediate face densities. In this regime, the second-largest eigenvalue modulus, the inverse participation ratio of the stationary distribution, and the mean first-passage time all exhibit non-monotonic behavior, reaching their largest values at intermediate face densities. The corresponding first-passage-time distributions reveal an enhanced probability of unusually long trajectories. Together, these results establish a simple framework for investigating how heterogeneous higher-order connectivity reshapes spectral and transport properties beyond pairwise network dynamics.
\end{abstract}

\maketitle

\section{Introduction}
The study of stochastic processes on topologically nontrivial spaces has attracted increasing attention in statistical physics, motivated by applications ranging from diffusion on complex networks and higher-order structures to relaxation in topological systems and transport in constrained geometries \cite{Masuda2017,RiascosJCompNetworks2021,SawadaPhysRevLett2024}. Among these processes, \emph{discrete-time random walks} (DTRW), on either regular or complex networks, has established a firm theoretical framework to model the stochastic journey of a system (the walker), over its accessible states (the network nodes), transiting among them via connected pairs of the system states (network edges between nodes). Recently, higher-order structures have emerged as a necessary ingredient for describing the dynamics of complex systems beyond pairwise interactions \cite{BensonScience2016,LambiotteNatPhys2019,BattistonNatPhys2021,MajhiJRoySocInterface2022,MillanNatPhys2025,WangPhyRep2024,IacopiniNature2019}. Their incorporation into stochastic dynamics is therefore of growing interest. A natural extension is to replace edges with hyperedges, which can connect more than two nodes, thereby generalizing the pairwise transition rates of random walks (RWs) on graphs to stochastic dynamics on hypergraphs \cite{CarlettiPhysRevE2020,TraversaPhysRevE2023,BensonSIAM2017,HayashiBook2020}. 

A related framework for incorporating higher-order structures into stochastic dynamics is provided by \emph{simplicial complexes}. From a modeling perspective, these provide the ideal framework for moving beyond graph-based diffusion by incorporating higher-order effects (described in terms of the elemental structures, as the simplices of simplicial complexes: edges, faces, and higher-dimensional simplices) and making topological constraints explicit \cite{BianconiBook2021}. Thus, random walks on simplicial complexes allow a powerful generalization of the stochastic journey on nodes (representing simple system states), to the case for which the walk occurs on  structured elements of the system represented by simplices of order $k$, and the transitions among these elements are induced by the incident simplices of order $k-1$ and $k+1$. 
In this paper, we consider a random walk on a set of edges (order-1 simplices), where transitions between neighboring edges are defined by the incident vertices (order-0 simplices) and triangles (order-2 simplices) of the edge under consideration. 

The main goal of this paper is to study the transport properties of DTRWs on finite spaces described by simplicial complexes 
whose dynamics explicitly incorporate the dimensionality of the simplices that compose such spaces. In this sense, we consider a higher-order random walk (HORW).  
In short, the transition between neighboring simplices of dimension $k$ is implemented through simplices of one dimension smaller, $k-1$, or one dimension larger $k+1$. We interpret its up-and-down components as mechanisms for transport across different dimensions \cite{FebbeArXiv2026a,FebbeArXiv2026b}. 
In contrast to these approaches, in which the walker's 
state itself resides on simplices of varying dimension (used, respectively, to rank simplices across orders and to optimize node-to-node navigation with teleportation), the dynamics considered here constrains 
the random walk to occur on the edge space $X_1$: the vertex and triangle channels act only as intermediate steps of a two-step composition, so that the observable of interest is transport within a fixed simplicial order, modulated by the incidence of higher-dimensional structure. Moreover, we consider a random ensemble of simplicial complexes with partially retained faces, rather than a fixed complex, which allows us to isolate the role of quenched structural 
heterogeneity, absent by construction in a deterministic complex, as the origin of the transport bottleneck reported below.

This paper is organized as follows: In Sect.~\ref{sec:prelim} we introduce the mathematical framework and define the effective edge-walk dynamics in terms of lower and upper transition channels. We then present the random face-diluted ensemble $X(n,p)$, which constitutes the main object of study, and discuss its two natural limiting regimes. After that, we examine the complete case $p=1$ as an exactly solvable benchmark. We next analyze the random ensemble through structural observables, spectral relaxation, stationary-state localization, and mean first-passage times. Finally in Sect.\ref{conclusions}, we summarize the main results and discuss possible extensions of the present framework.

\section{Mathematical framework}
\label{sec:prelim}

In this section we introduce the concepts and fix the notation needed to define a DTRW on the state space of $1$-simplices (edges) of a finite simplicial complex, which is the main subject matter of the paper. In this setting, the edges of a given set of vertices constitute the relevant state space, where transport can be regarded as a stochastic exploration of the space of interactions. This perspective becomes particularly important in systems in which interactions are not independent but are organized into higher-order structures. In such a case, the walker explores a space of pairwise interactions whose connectivity is constrained by the surrounding simplicial structure. The presence of triangles and higher-dimensional simplices can consequently modify the pathways available to the walker, giving rise to transport properties that cannot be inferred solely from the underlying network of vertices.

Throughout this paper, we consider simplicial complexes of dimension at most two, i.e., composed of vertices, edges and triangles. 

\subsection{Simplicial complexes}
\label{subsec:sc}

Let $\mathcal{V}$ be a finite set whose elements are called vertices. A $k$-simplex is a subset $\sigma\subset \mathcal{V}$ of size
$|\sigma|=k+1$ and a simplicial complex $X$, is a collection of simplices $\sigma$ such that if $\sigma\in X$
and $\tau\subset \sigma$, then $\tau\in X$.
We will focus our analysis on the case  $k=0,1,2$, in setting we write $X_k$ for the set of $k$-simplices in $X$, and denote $n_k=|X_k|$ the number of these in $X$ (see \cite{Schaub2020}).

A (simple) graph $(\mathcal{V},\mathcal{E})$ can be identified with a 1-dimensional simplicial complex $X_0=\mathcal{V}$ (vertices), $X_1=\mathcal{E}$ (edges), and $X_2=\emptyset$ (triangles).
In a general 2-dimensional simplicial complex, $X_2$ contains $2$-simplices, which
encode higher-order relations among triplets of vertices. 

We proceed now to specify the concept of \emph{incidence}, which specify the relation between the elements in a simplicial complex and allows the definition of a DTRW on it. Incidence is understood here in the standard way: a vertex $v\in X_0$ is incident to an edge $e\in X_1$
if $v\in e$, and an edge $e\in X_1$ is incident to a triangle $f\in X_2$ if $e\subset f$. 

\subsection{Unsigned incidence structure}
Let $\mathbf B_1\in\{0,1\}^{n_0\times n_1}$, $\mathbf B_2\in\{0,1\}^{n_1\times n_2}$ denote the unsigned vertex--edge and edge--triangle incidence matrices, respectively. These are given explicitly by
\begin{subequations}
\begin{align}
\bigl(\mathbf B_1\bigr)_{ve}&=
\begin{cases}
1, & v\subset e,\\
0, & \text{otherwise,}
\end{cases}\\
\bigl(\mathbf B_2\bigr)_{ef}&=
\begin{cases}
1, & e\subset f,\\
0, & \text{otherwise.}
\end{cases}
\end{align}    
\end{subequations}

\begin{subequations}
We define the degree of the vertex $v$ as
\begin{equation}
d_V(v)=\sum_{e\in X_1} (\mathbf B_1)_{ve},
\end{equation}
which gives the number of edges incident to vertex $v$,
and the triangle degree of an edge $e$ as
\begin{equation}
d_T(e)=\sum_{f\in X_2} (\mathbf B_2)_{ef},
\end{equation}
that is, the number of triangles containing $e$.
\end{subequations}

The vertex and triangle degrees allow to introduce the diagonal matrices \begin{subequations}
\begin{equation}
D_V=\mathrm{diag}\big(d_V(v)\big)_{v\in X_0},
\end{equation}
\begin{equation}
D_T=\mathrm{diag}\big(d_T(e)\big)_{e\in X_1}.
\end{equation}
\end{subequations}
For the specific cases considered here 
every vertex belongs to at least one edge, thus $D_V^{-1}$ is well defined. 
However, the number of incident triangles to some edges may be zero, thus $d_T(e)=0$ for some $e\in X_1$. This case is resolved by the introduction of the diagonal matrix
\begin{equation}
\widetilde D_T^{-1}
=
\mathrm{diag}\!\left(\widetilde d_T(e)^{-1}\right)_{e\in X_1},
\end{equation}
with
\begin{equation}
\widetilde d_T(e)^{-1}
=
\begin{cases}
\dfrac{1}{d_T(e)}, & d_T(e)>0,\\[2mm]
0, & d_T(e)=0.
\end{cases}
\end{equation}
We also introduce the \emph{availability matrix}
\begin{equation}
A=\mathrm{diag}(a_e)_{e\in X_1},
\qquad
a_e=
\begin{cases}
1, & d_T(e)>0,\\
0, & d_T(e)=0
\end{cases},
\end{equation}
which identifies the edges that belong to at least one triangle. This matrix will be used below to encode the availability of the upper transition channel in the effective dynamics.

\section{The Model}
\label{sec:model}
\subsection{Edge-state dynamics}
We consider a discrete-time Markov chain whose state space corresponds to the set of 1-simplices (edges), $X_1$, of a finite simplicial complex $X$ of dimension at most two. Thus, if $n_1=|X_1|$, the state of the walker at time $t$ is given by the column probability vector $p(t)\in \mathbb{R}^{n_1}$ whose entries $p_e(t)\ge 0$ indicate the probability that the walker occupies the edge $e$ and satisfies normalization, i.e.,  $\sum_{e\in X_1} p_e(t)=1$.

The evolution of $p(t)$ is dictated by the Markov chain master equation, which in the column-vector convention is written as
\begin{equation}
\label{RW-MasterEq}
p(t+1)=P_{\mathrm{eff}}\,p(t),
\end{equation}
where $P_{\mathrm{eff}}$ is the matrix that specifies the transition between edges in $X_1$.

\subsection{Two-step transition mechanism}
The core approach in our analysis is that the transition between two given edges is not prescribed directly but rather is built in two steps. The walker moves from one edge to another edge through an intermediate simplex of adjacent dimension. Starting from an edge $e\in X_1$, the walker may first select either one of its incident vertices or one of its incident triangles, and then choose an edge incident to that intermediate simplex. This defines two elementary transport channels: a lower channel mediated by vertices and an upper channel mediated by triangles.

Throughout the paper, we use the convention
\begin{equation*}
 P_{ij}:C_i\to C_j,
\qquad i,j\in\{0,1,2\},
\end{equation*}
to denote the transition operator from the chain space $C_i$ (the space of all possible configurations of a physical quantity distributed over the simplices of dimension $i$), to the chain $C_j$. Thus, the first index indicates the origin space and the second index the target space. In the case considered here, $C_0$, $C_1$, and $C_2$ denote the chain spaces associated with vertices, edges, and triangles, respectively. 
Since probability states are represented by column vectors, operator compositions act on a given vector from right to left. Therefore, the transition probability between edges through vertices
(transition by lowering) is represented by
$P^\downarrow=P_{01}P_{10}$,
corresponding to the sequence
$C_1 \xrightarrow{P_{10}} C_0 \xrightarrow{P_{01}} C_1$,
whereas the transition between edges mediated by triangles
(transition by rising) is 
$P^\uparrow=P_{21}P_{12}$,
corresponding to
$C_1 \xrightarrow{P_{12}} C_2 \xrightarrow{P_{21}} C_1$.

\subsection{Elementary transition operators}

We now define the elementary operators involved in the two-step composition of the transition dynamics.

The transition from edges to vertices is implemented by uniformly choosing any of the two endpoints of the given edge, we make the identification
\begin{equation}
P_{10}=\frac12\,\mathbf B_1,
\end{equation}
or equivalently,
\begin{equation}
(P_{10})_{ve}
=
\begin{cases}
\dfrac12, & v\subset e,\\
0, & \text{otherwise.}
\end{cases}
\end{equation}
The reverse transition, i.e., from vertices to edges, is uniform among all edges incident to a given vertex, i.e.
\begin{equation}
P_{01}=\mathbf B_1^\top D_V^{-1},
\end{equation}
that is,
\begin{equation}
(P_{01})_{ev}
=
\begin{cases}
\dfrac{1}{d_V(v)}, & v\subset e,\\
0, & \text{otherwise.}
\end{cases}
\end{equation}
Similarly, the transition from edges to triangles is uniform among all triangles incident to a given edge:
\begin{equation}
P_{12}=\mathbf B_2^\top \widetilde D_T^{-1},
\end{equation}
so that
\begin{equation}
(P_{12})_{fe}
=
\begin{cases}
\dfrac{1}{d_T(e)}, & e\subset f,\ \ d_T(e)>0,\\
0, & \text{otherwise.}
\end{cases}
\end{equation}

Finally, the transition from triangles to edges is uniform among the three boundary edges of a triangle:
\begin{equation}
P_{21}=\frac13\,\mathbf B_2,
\end{equation}
namely,
\begin{equation}
(P_{21})_{ef}
=
\begin{cases}
\dfrac13, & e\subset f,\\
0, & \text{otherwise.}
\end{cases}
\end{equation}

\subsection{Lower and upper edge walks}

By composing the elementary operators, we obtain two edge-to-edge transition kernels. The transition by lowering to vertices is defined by
\begin{equation}
P^\downarrow
=
P_{01}P_{10}
=
\frac12\,\mathbf B_1^\top D_V^{-1} \mathbf B_1,
\end{equation}
whose entries can be expressed as
\begin{equation*}
P^\downarrow_{e'e}
=
\frac12\sum_{v\in X_0}
\frac{\mathbf 1_{\{v\subset e\}}\mathbf 1_{\{v\subset e'\}}}{d_V(v)},
\end{equation*}
where $\mathbf 1_{\{v\subset e\}}$ equals 1 if $v\subset e$ and 0 otherwise (and similarly for 
$\mathbf 1_{\{v\subset e'\}}$). 
Thus, $P^\downarrow_{e'e}>0$ precisely when the edges $e$ and $e'$ share a vertex, with the diagonal term corresponding to the possibility of returning to the same edge after visiting one of its endpoints.

The upper walk, by going through triangles, is defined by
\begin{equation}
P^\uparrow
=
P_{21}P_{12}
=
\frac13\,\mathbf B_2\,\mathbf B_2^\top \widetilde D_T^{-1}.
\end{equation}
Its entries are
\begin{equation*}
P^\uparrow_{e'e}
=
\frac13\sum_{f\in X_2}
\frac{\mathbf 1_{\{e\subset f\}}\mathbf 1_{\{e'\subset f\}}}{d_T(e)},
\end{equation*}
with the convention that the whole column is zero when $d_T(e)=0$. Hence $P^\uparrow_{e'e}>0$ precisely when $e$ and $e'$ belong to a common triangle.

Notice that, in a $2$-dimensional simplicial complex, two edges belonging to a common triangle necessarily share a vertex. Therefore, the upper channel does not create new edge adjacencies beyond those already present in the line graph of the $1$-skeleton; rather, it modifies the transition weights by incorporating higher-order incidence information.

\begin{figure}
\includegraphics[width=\columnwidth, trim=20 150 20 150, clip=true]{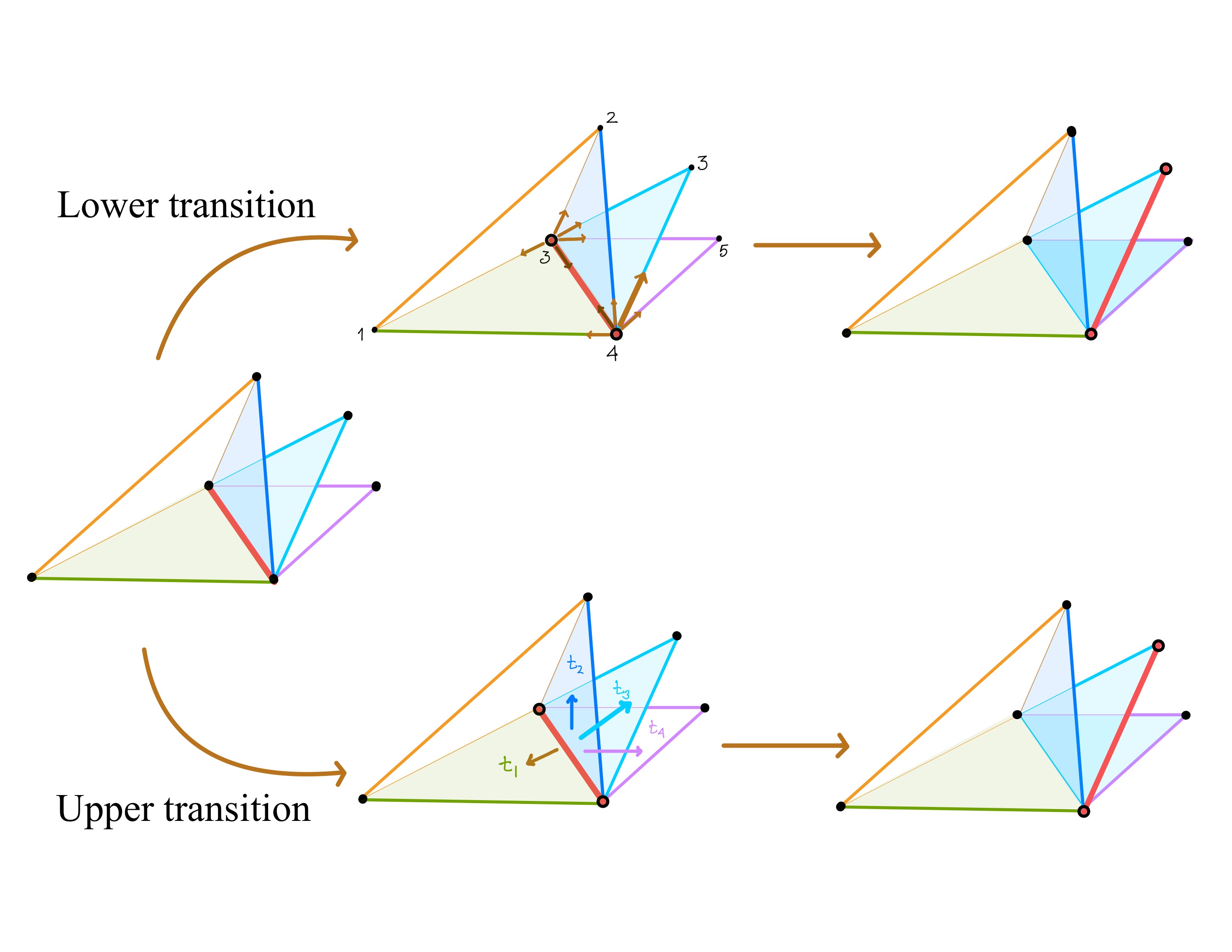}
\caption{Schematic representation of the two elementary transport channels defining the edge walk. Lower transition: the walker moves from an edge (red) to one of its two incident vertices $\{3,4\}$, in the example the walker chose vertex 4, then chooses one edge incident to the chosen vertex, edge $(4,3)$ in the example. Upper transition: the walker moves from an edge (red) to one of its incident triangles and then to another boundary edge of that triangle.}
\label{fig:placeholder}
\end{figure}

\subsection{The effective transition operator}

We combine both channels into a single effective transition matrix. Let $q\in[0,1]$ denote the weight assigned to the lower channel whenever the upper channel is available. We define
\begin{equation}
\label{eq:Peff}
P_{\mathrm{eff}}(q)
=
P^\downarrow\!\big[I-(1-q)A\big]
+
(1-q)\,P^\uparrow A .
\end{equation}
Here, $A$ acts as a mask that restricts the upper channel to edges with $d_T(e)>0$. By construction, $P_{\mathrm{eff}}(q)$ is column-stochastic for every $q\in[0,1]$ (a short proof is given in Appendix~A). Equivalently, column by column,
\[
{P_{\mathrm{eff}}}_{(\cdot,e)}=
\begin{cases}
q\,P^\downarrow_{(\cdot,e)}+(1-q)\,P^\uparrow_{(\cdot,e)},
& d_T(e)>0,\\[2mm]
P^\downarrow_{(\cdot,e)},
& d_T(e)=0.
\end{cases}
\]

This definition has two useful consequences. First, it remains well defined even when some edges do not belong to any triangle. Second, in the purely graph regime $X_2=\varnothing$, one has $A=0$ and therefore
\begin{equation*}
P_{\mathrm{eff}}(q)=P^\downarrow,
\end{equation*}
independently of $q$. In particular, for the face-diluted complexes studied below, the limit $p=0$ reduces exactly to the graph-based edge walk.

Notice that the effective walk always has self-loops. Indeed,
\begin{equation*}
P^\downarrow_{ee}
=
\frac12\sum_{v\subset e}\frac{1}{d_V(v)}>0
\qquad \forall e\in X_1.
\end{equation*}
Therefore,
\begin{equation*}
{P_{\mathrm{eff}}}_{(e,e)}>0
\qquad \forall e\in X_1.
\end{equation*}
As a consequence, once irreducibility holds, the chain is automatically aperiodic.

\section{Random face-diluted complete complexes}
\label{sec:random-ensemble}

Having defined the general edge-walk dynamics, we now turn to explore the random face-diluted ensemble $X(n,p)$, that is, the case for which only a fraction $p$ of the total number of possible triangles in the simplicial complex is present. This constitutes the main object of study in this work. Such a construction interpolates, in a controlled way, between the purely graph-based (no triangles) and the fully simplicial regimes (all possible triangles are present) by randomly diluting the set of triangular faces while keeping the complete $1$-skeleton, in analogy with the classical random-graph ensemble \cite{ErdosRenyi1959}. The fraction of triangles present in the simplicial complex will characterize the transport properties of the RWs considered.   

We fix a set of $n$ vertices,
\begin{equation*}
X_0=\{1,\dots,n\},
\end{equation*}
and keep the complete graph $K_n$ as the $1$-skeleton. Thus, all possible edges are always present, so that
\begin{equation*}
X_1=\big\{\{i,j\}:1\le i<j\le n\big\},
\qquad
|X_1|=\binom{n}{2}.
\end{equation*}
Randomness is introduced only at the level of faces: each possible triangle
\begin{equation*}
\{i,j,k\},
\qquad
1\le i<j<k\le n,
\end{equation*}
is independently included in the simplicial complex with probability $p\in[0,1]$, in the spirit of the Linial--Meshulam random 2-complex model \cite{LinialMeshulam2006} (see \cite{Kahle2014Survey} for a survey). Therefore,
\begin{equation*}
|X_2|\le \binom{n}{3},
\end{equation*}
and the parameter $p$ controls the density of $2$-simplices in the complex.

This ensemble has two natural limiting cases. When $p=0$, all faces are removed and the complex reduces to the complete graph $K_n$, so that the effective walk becomes purely lower:
\begin{equation*}
P_{\mathrm{eff}}(q)=P^\downarrow.
\end{equation*}
At the opposite extreme, when $p=1$, all triangles are present and one recovers the complete $2$-dimensional simplicial complex, namely the complete $2$-skeleton of the $(n-1)$-simplex. Thus, by varying $p$ on can interpolate between a purely graph-based edge walk and a higher-order edge walk influenced by triangular incidences.

For a given realization of the ensemble, the upper degree of an edge,
\begin{equation*}
d_T(e)=|\{f\in X_2:\ e\subset f\}|,
\end{equation*}
becomes a random quantity. Consequently, the set of active edges 
\begin{equation*}
X_1^+=\{e\in X_1:\ d_T(e)>0\},
\end{equation*}
as well as the operators $P^\uparrow$ and $P_{\mathrm{eff}}$, depend on the particular realization of $X(n,p)$. By contrast, the lower walk $P^\downarrow$ is deterministic in this ensemble, since it depends only on the complete $1$-skeleton. In this sense, random face removal does not modify the graph support of the walk, but it does alter the higher-order contribution to the transition probabilities through the random upper degrees.

This ensemble provides a controlled framework to isolate the effect of higher-order structure on transport. Since the underlying graph remains fixed, all changes in the effective dynamics are produced solely by the random availability of triangular faces and by the competition between the lower and upper channels.

In the numerical analysis presented below, independent realizations of $X(n,p)$ are generated for each fixed value of $p$, the corresponding effective transition matrix $P_{\mathrm{eff}}$ is constructed, and the observables of interest are computed. Unless otherwise stated, all reported quantities are averaged over independent realizations at fixed $(n,p,q)$.

\section{The complete case ($p=1$) as an exactly solvable benchmark}
As stated in the previous section, the random face-diluted ensemble 
attains two natural limiting cases: the graph-based regime $p=0$ and the fully simplicial regime $p=1$. Here we focus on the latter. When $p=1$, all triangular faces are present and the complex becomes the complete $2$-skeleton of the $(n-1)$-simplex. Because of its maximal symmetry, this limit can be analyzed exactly and provides a natural benchmark for the higher-order random walk introduced above.

We therefore consider the simplicial complex formed by $n$ vertices, all $\binom{n}{2}$ possible edges, and all $\binom{n}{3}$ possible triangles. Explicitly,
\begin{equation*}
X_0=\{1,\dots,n\}, \qquad |X_0|=n,
\end{equation*}
\begin{equation*} X_1=\big\{\{i,j\}:1\le i<j\le n\big\}, \qquad |X_1|=\binom{n}{2}, \end{equation*} and \begin{equation*} X_2=\big\{\{i,j,k\}:1\le i<j<k\le n\big\}, \qquad |X_2|=\binom{n}{3}. \end{equation*}
Its $1$-skeleton is the complete graph $K_n$. In particular, every vertex has degree $d_V(v)=n-1$, every edge belongs to exactly $d_T(e)=n-2$ triangles, and therefore every edge is active. Hence $A=I$, so that the effective transition matrix reduces to
\begin{equation}
P_{\mathrm{eff}}(q)=qP^\downarrow+(1-q)P^\uparrow.
\end{equation}

For the lower walk, the general expression
\begin{equation*}
P^\downarrow_{e'e}
=
\frac12\sum_{v\in X_0}
\frac{\mathbf 1_{\{v\subset e\}}\mathbf 1_{\{v\subset e'\}}}{d_V(v)}
\end{equation*}
gives
\begin{equation}
P^\downarrow_{e'e}
=
\begin{cases}
\dfrac{1}{n-1}, & e'=e,\\[2mm]
\dfrac{1}{2(n-1)}, & e'\neq e \text{ and } e,e' \text{ share a vertex},\\[2mm]
0, & \text{otherwise}.
\end{cases}
\end{equation}
Indeed, if $e'=e$, the sum receives the contribution $1/[2(n-1)]$ from each endpoint of $e$, whereas if $e\neq e'$ and the two edges share a vertex, there is exactly one nonzero contribution of the same form.

For the upper walk, the corresponding formula
\begin{equation*}
P^\uparrow_{e'e}
=
\frac13\sum_{f\in X_2}
\frac{\mathbf 1_{\{e\subset f\}}\mathbf 1_{\{e'\subset f\}}}{d_T(e)}
\end{equation*}
yields.
For edges sharing a common triangle, write $e \sim e'$; i.e., $e\sim e'$ 
if there exists $f\in X_2$ with $e,e'\subset f$. Then
\begin{equation}
P^\uparrow_{e'e}
=
\begin{cases}
\dfrac{1}{3}, & e'=e,\\[2mm]
\dfrac{1}{3(n-2)}, & e'\neq e,\ e\sim e',\\[2mm]
0, & \text{otherwise}.
\end{cases}
\end{equation}
Here, if $e'=e$, the sum runs over the $n-2$ triangles incident to $e$, while if $e\neq e'$ and the two edges share a vertex, there is exactly one triangle containing both. Thus, in the complete case, $P^\downarrow$ and $P^\uparrow$ have the same support: both connect pairs of edges that share a vertex, but they assign different transition weights.
\subsection{The 2-skeleton of the 5-simplex}
We now consider the complete 2-skeleton of a set of $n=6$ vertices, i.e., the 2-skeleton of the 5-simplex. This example makes the general construction fully explicit and serves as a convenient benchmark for the numerical analysis developed below.

We further fix $q=\frac{1}{2}$ so that both transport channels contribute equally, and the effective transition among edges takes the explicit form
\begin{equation}
P_{\mathrm{eff}}=\frac12\left(P^\downarrow+P^\uparrow\right),
\end{equation}
that is,
\begin{equation*}
(P_{\mathrm{eff}})_{e'e}
=
\begin{cases}
\dfrac{4}{15}, & e'=e,\\[2mm]
\dfrac{11}{120}, & e'\neq e \text{ and } e,e' \text{ share a vertex},\\[2mm]
0, & \text{otherwise}.
\end{cases}
\end{equation*}
Since each edge of $K_6$ shares a vertex with exactly eight other edges, the column sums to one 
as required. Moreover, by symmetry, $P_{\mathrm{eff}}$ is doubly stochastic, and therefore its stationary distribution is uniform:
\begin{equation*}
\pi_e=\frac{1}{15},\qquad e\in X_1.
\end{equation*}
The spectrum of $P_{\mathrm{eff}}$ can also be obtained exactly. Since the support of the edge walk is the graph whose vertices are the $15$ edges of $K_6$ and whose adjacency is determined by shared endpoints, the relevant support graph is the line graph of $K_6$. Denoting by $A_{\mathrm{L}}$ its adjacency matrix, one may write
\begin{equation}
P_{\mathrm{eff}}=\frac{4}{15}I+\frac{11}{120}A_{\mathrm{L}}.
\end{equation}
where $A_{\mathrm{L}}$ is the adjacency matrix of the line graph. Using the known spectrum of $T(6)$ \cite{brouwer2011spectra,Cvetkovic2004}, the eigenvalues of $P_{\mathrm{eff}}$ are
\begin{equation*}
\lambda_1=1,\qquad
\lambda_2=\frac{9}{20},\qquad
\lambda_3=\frac{1}{12},
\end{equation*}
with multiplicities $1$, $5$, and $9$, respectively. Since the kernel is symmetric, the spectral gap is
\begin{equation*}
\gamma=1-\lambda_2=\frac{11}{20}.
\end{equation*}

This exact spectrum shows that the complete case is strongly mixing. In particular, all non-stationary modes decay at rates controlled by eigenvalues strictly smaller than one, and the leading relaxation scale is governed by $\lambda_2=9/20$. 
The role of the upper channel in this regime is not to enlarge the support of the walk, but rather to redistribute the transition weights uniformly across the edge space. This maximally regular situation serves as the natural baseline for the face-diluted complexes analyzed below.

\section{Numerical results}
\subsection{Structural properties}
Before analyzing the effective walk \eqref{RW-MasterEq}, it is useful to characterize how face dilution affects the higher-order structure of the underlying random simplicial complex. In the ensemble $X(n,p)$, the key structural quantity is the triangle degree of an edge, $d_T(e)$, namely the number of $2$-simplices incident to $e$.

Each edge of the complete graph $K_n$ can belong to exactly $n-2$ distinct triangles. Since each triangle in $X(n,p)$ is retained with probability $p$, the random variable $d_T(e)$ follows a binomial law,
\begin{equation*}
d_T(e)\sim \mathrm{Binomial}(n-2,p).
\end{equation*}
Therefore, its mean and variance are given by
\begin{equation}
\langle d_T\rangle = (n-2)p,
\label{dt}
\end{equation}
and
\begin{equation*}
\mathrm{Var}(d_T)=(n-2)p(1-p),
\end{equation*}
respectively. In particular, the standard deviation reads
\begin{equation}
\sigma(d_T)=\sqrt{(n-2)p(1-p)}.
\label{sdt}
\end{equation}

Another quantity of direct interest is the fraction of inactive edges,
\begin{equation*}
\phi_0=\frac{1}{|X_1|}\sum_{e\in X_1}\mathbf{1}_{{d_T(e)=0}},
\end{equation*}
which measures the proportion of edges for which the upper channel is unavailable. Since $d_T(e)=0$ precisely when none of the $n-2$ possible triangles containing $e$ is present, one has
\begin{equation}
\langle \phi_0\rangle = (1-p)^{n-2}.
\label{fi0}
\end{equation}

These structural observables characterize complementary aspects of the
higher-order connectivity created by face dilution. The mean triangle degree
determines the typical availability of the upper channel, whereas its
standard deviation quantifies the degree of structural heterogeneity across
edges. In contrast, the inactive-edge fraction measures the prevalence of
edges for which the upper channel is completely absent. Together, these
quantities provide a structural description of the simplicial ensemble that
will serve as the basis for interpreting the dynamical behavior discussed in
the following sections.

\begin{figure}
\centering
\includegraphics[trim=0.4cm 13.cm 4cm 0.8cm,clip=true, width=1.1\linewidth]{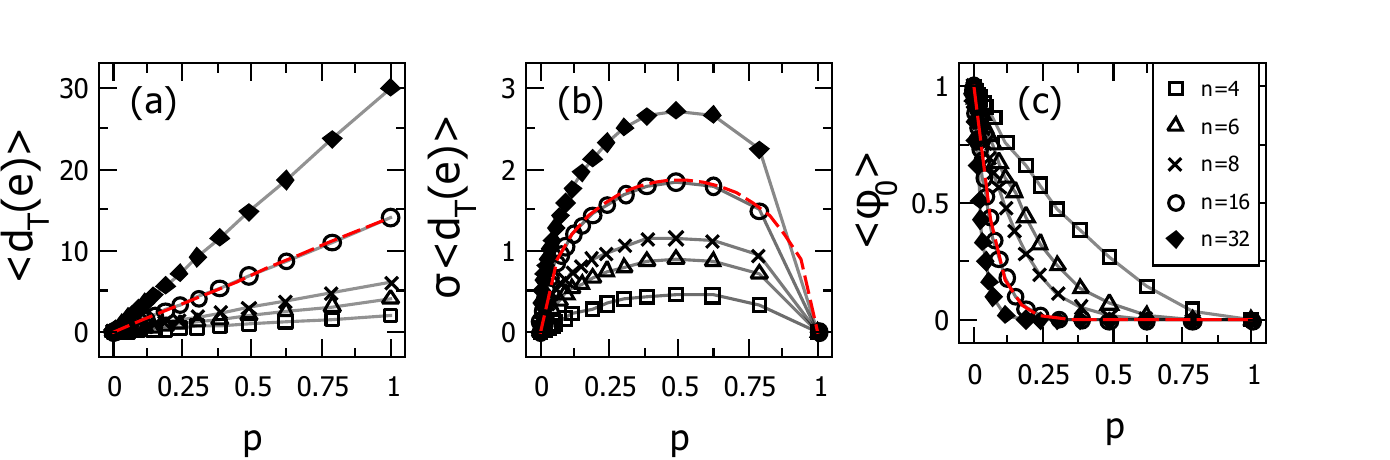}
\caption{Structural properties of the random face-diluted ensemble $X(n,p)$. (a) Mean triangle degree $\langle d_T\rangle$, (b) standard deviation $\sigma(d_T)$, and (c) inactive-edge fraction $\phi_0$, as functions of the face-retention probability $p$. Symbols correspond to different values of $n$, while the dashed red line indicates the theoretical prediction for $n=16$ [Eqs.~(\ref{dt})--(\ref{fi0})].}
\label{Fig1}
\end{figure}

As expected, these structural observables depend only on the random
simplicial ensemble and are therefore independent of the mixing parameter
$q$, which enters exclusively through the definition of
$P_{\mathrm{eff}}$. More importantly, Fig.~\ref{Fig1} shows that face
dilution does not simply reduce the amount of higher-order connectivity.
Instead, it generates an intermediate regime in which the upper channel is
highly heterogeneous: the average triangle degree increases monotonically,
whereas its fluctuations become maximal at intermediate values of $p$ before
decreasing again as the simplicial complex approaches the complete
2-skeleton. As we show below, this structurally heterogeneous regime
produces distinctive spectral and dynamical signatures, including slow
relaxation, localization of the stationary state, and inefficient transport.

\subsection{Spectral properties}
Having characterized the structural heterogeneity introduced by face dilution, we next examine its  consequences on the spectrum through the effective transition operator. Since $P_{\mathrm{eff}}$ is a finite column-stochastic matrix, it always has eigenvalue $1$, corresponding to the stationary state. Given that $P_{\mathrm{eff}}$ is generally non-symmetric for face-diluted realizations, we characterize the slowest relaxation mode through the second-largest eigenvalue modulus,
\begin{equation}
\lambda_*=\max_{\lambda\in\sigma(P_{\mathrm{eff}})\setminus{1}} |\lambda|.
\end{equation}
Values of $\lambda_*$ close to unity indicate slow relaxation towards stationarity, whereas smaller values correspond to faster mixing.

\begin{figure}
\centering
\includegraphics[trim={0.45cm 5.3cm 4cm 0},clip, width=1.1\linewidth]{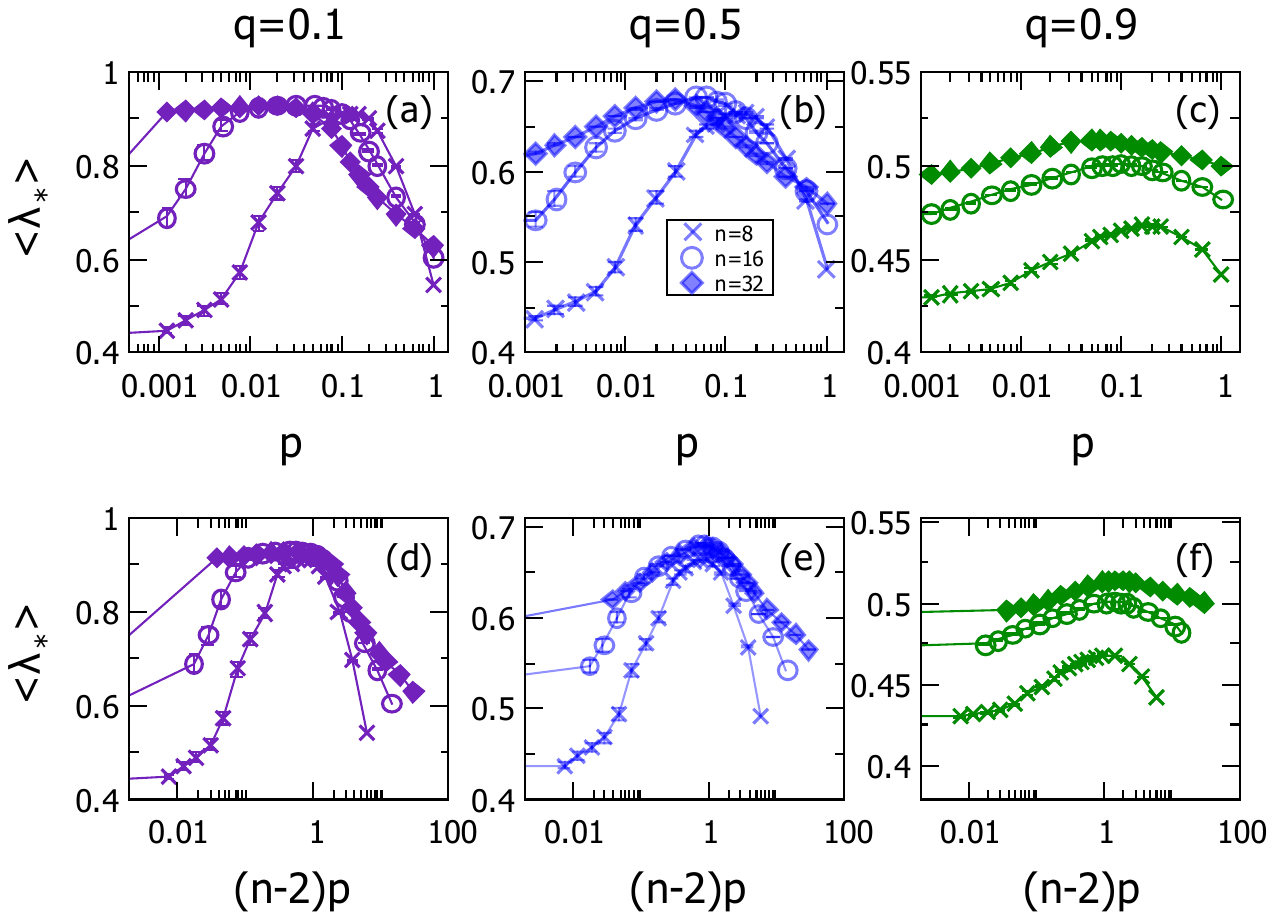}
\caption{
Average second-largest eigenvalue modulus $\langle \lambda_* \rangle$ of the effective transition operator. The top row presents $\langle \lambda_* \rangle$ as a function of the face-retention probability $p$, whereas the bottom row shows the corresponding results as a function of the mean upper degree $\langle d_T\rangle=(n-2)p$. Columns correspond to $q=0.1$, $0.5$, and $0.9$, respectively.  Each symbol represents the average over an ensemble of $180000/n$ independent realizations.
}
\label{Fig3b}
\end{figure}

Figure~\ref{Fig3b} shows the average of the second-largest eigenvalue modulus, $\langle \lambda_* \rangle$, over the ensembles of random 2-skeleton for system sizes $n=8$, $16$, and $32$. The top row displays $\langle \lambda_* \rangle$ as a function of the face-retention probability $p$, while the bottom row presents the same results as a function of the mean upper degree $\langle d_T\rangle=(n-2)p$. The case $q=1$ is not shown because in that limit the dynamics reduces to the purely lower walk and $\langle \lambda_* \rangle$ becomes independent of $p$. For all values of $q<1$, $\langle \lambda_* \rangle$ exhibits a pronounced non-monotonic dependence, reaching a maximum at intermediate values of the control parameter. Thus, the slowest relaxation does not occur in either the purely graph-based or the fully simplicial limit, but rather in an intermediate regime where the upper channel is only partially available. Although the location of the maximum shifts with system size when represented as a function of $p$, this dependence is substantially reduced when the results are expressed in terms of the mean upper degree, indicating that the slowdown is primarily controlled by the average upper connectivity rather than by the face-retention probability alone.

This behavior naturally suggests examining the stationary state itself to investigate how the competition between lower transport and partial higher-order connectivity redistributes stationary weight across the edge space.

We now examine how face dilution affects the stationary state of the effective walk. For each realization, we compute the stationary distribution
$\pi$ of $P_{\mathrm{eff}}$ and quantify its localization through
the inverse participation ratio (IPR)\cite{Wegner1980}, a standard measure of localization also used to characterize eigenvector concentration in networks \cite{MartinPRE2014,MartinezMartinez2019Entropy,MendezBermudez2015},
\begin{equation}
\mathrm{IPR}(\pi)=\sum_{e\in X_1}\pi_e^2,
\end{equation}
where $\pi_e$ denotes the stationary probability of occupying the edge $e$. 
We then average this quantity over independent realizations at fixed
$(n,p,q)$, obtaining
$\langle\mathrm{IPR}(\pi)\rangle$. Large values of
$\langle\mathrm{IPR}(\pi)\rangle$
indicate that the stationary distribution is concentrated on a
small subset of edges, whereas values close to
$1/|X_1|$
correspond to a more delocalized stationary state, close to the uniform distribution. Localized stationary distributions indicate poor transport properties. 

Figure~\ref{Fig:IPR} shows the average inverse participation ratio $\langle \mathrm{IPR}(\pi)\rangle$ of the stationary distribution as a function of the face-retention probability $p$ for a system of size $n=16$. Panels (a)--(c) correspond to the representative values $q=0.1$, $0.5$, and $0.9$, respectively. For all values of $q$, $\langle \mathrm{IPR}(\pi)\rangle$ exhibits a pronounced non-monotonic dependence on $p$, reaching a maximum at intermediate face densities. Thus,for fix $q$ the stationary distribution becomes more localized
precisely in the regime where the upper connectivity is most
heterogeneous. For small values of $q$ (see $q=0.1$ in Fig.~\ref{Fig:IPR}), the walker prefers to visit edges via triangles, but these are scarce for small $p$, thus the walker gets trapped in some set of edges leading to relatively large values of the IPR. As $q$ increases, the overall magnitude of the IPR decreases, showing that increasing the contribution of the lower transition channel progressively suppresses localization and promotes a more delocalized stationary distribution over the edge set. The horizontal dashed line indicates the uniform reference value $1/|X_1|=1/120$.

\begin{figure}
\centering
\includegraphics[trim={0.45cm 12cm 4cm 0},clip, width=1\linewidth]{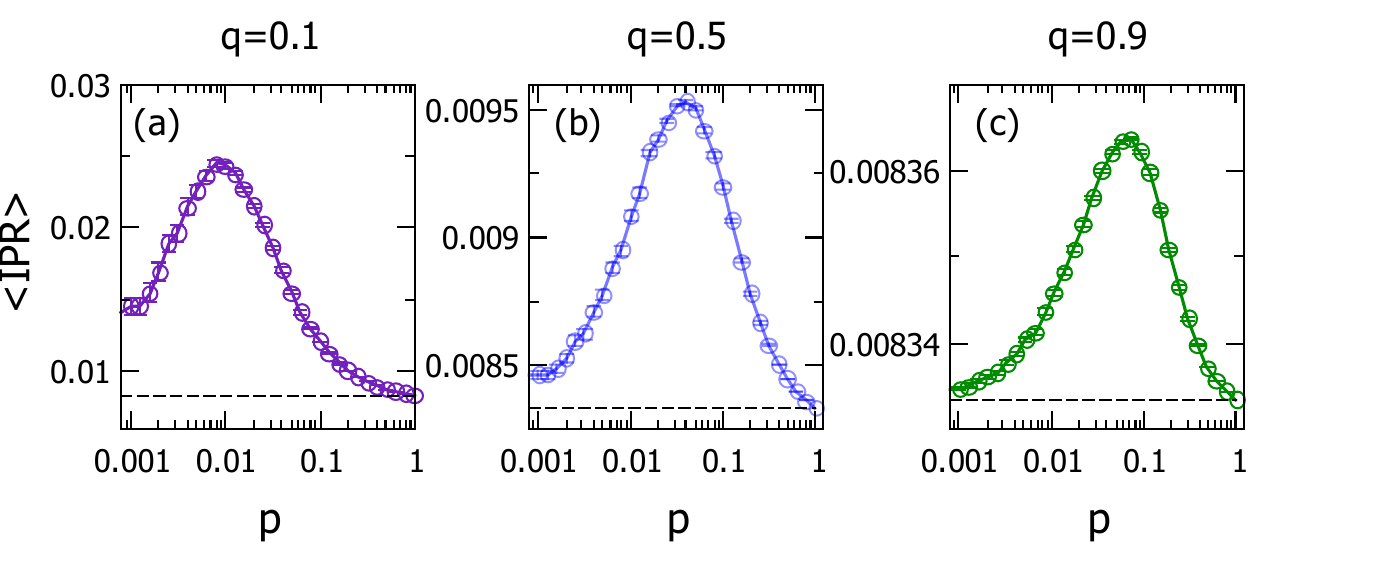}
\caption{
Average inverse participation ratio $\langle \mathrm{IPR}(\pi)\rangle$ of the stationary distribution as a function of the face-retention probability $p$ for a system of size $n=16$. Panels (a)–(c) correspond to $q=0.1$, $0.5$, and $0.9$, respectively. The horizontal dashed line denotes the uniform reference value $1/|X_1|=1/120$. Error bars represent the standard error over independent realizations.}
\label{Fig:IPR}
\end{figure}

\subsection{Dynamics}

The spectral and stationary-state analyses consistently identify an intermediate regime of hindered dynamics. We now investigate whether this same regime also manifests itself in transport observables.

For two distinct edges $e,e'\in X_1$, let
$T_{e\to e'}$ denote the first-passage time from $e$ to $e'$\cite{Redner2001,Condamin2007Nature}, i.e., the random number of steps required for a walker starting from $e$ to reach $e'$ for the first time. The corresponding
mean first-passage time is
\begin{equation}
\mathrm{MFPT}_{e\to e'}
=
\mathbb{E}\!\left[T_{e\to e'}\right].
\end{equation}

As a global measure of transport efficiency \cite{NohRieger2004,Tejedor2009GMFPT}, for each realization, we average over all ordered pairs of distinct edges,
\begin{equation}
\mathrm{MFPT}
=
\frac{1}{|X_1|(|X_1|-1)}
\sum_{\substack{e,e'\in X_1\\ e\neq e'}}
\mathrm{MFPT}_{e\to e'}.
\end{equation}
We then average this quantity over independent configurations of triangles considered in the simplicial complex at fixed
$(n,p,q)$, obtaining $\langle\mathrm{MFPT}\rangle$. Large values of
$\langle\mathrm{MFPT}\rangle$ indicate slower transport and therefore
less efficient exploration of the edge space, whereas smaller values
correspond to faster transport.

Figure~\ref{MFPT} (a) shows $\langle\mathrm{MFPT}\rangle$ as a function of the face-retention probability $p$ for several values of the mixing parameter $q$, at fixed system size $n=16$. The results exhibit a clear non-monotonic dependence on $p$ whenever the upper channel contributes to the dynamics. In particular, for $q<1$ the mean first-passage time develops a pronounced maximum at intermediate values of $p$, indicating that transport is slowest in a partially face-diluted regime. As stated before for the IPR, localized stationary distributions observed at small values of $q$, largely impacts the transport properties of the RW. 

\begin{figure*}
\centering
\includegraphics[trim={0cm 9.1cm 14cm 1.2cm},clip, height=4cm]{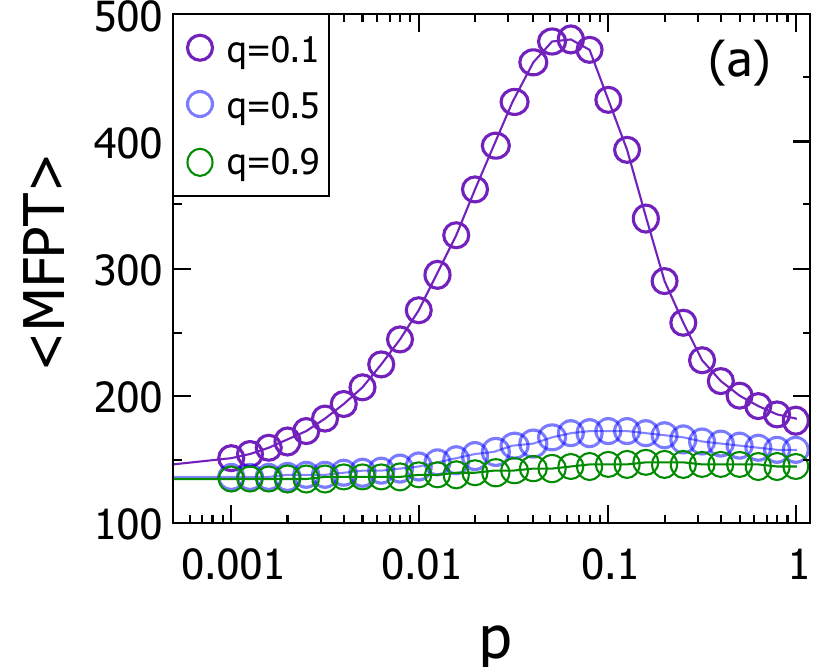}
\includegraphics[height=4cm]{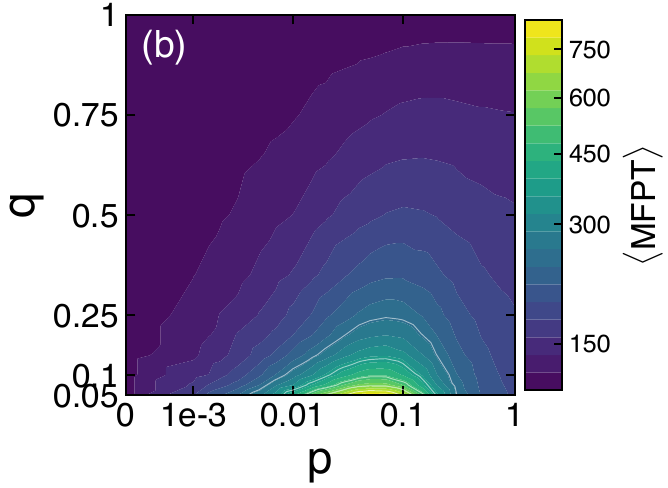}
\caption{
Mean first-passage time $\langle\mathrm{MFPT}\rangle$ for the effective random walk.
(a) Mean first-passage time $\langle\mathrm{MFPT}\rangle$ as a function of the face-retention probability $p$ for representative values of the mixing parameter $q$ and fixed system size $n=16$. 
(b) Contour map of $\langle\mathrm{MFPT}\rangle$ in the $(p,q)$ parameter space. 
}
\label{MFPT}
\end{figure*}

These transport results are fully consistent with the spectral and stationary-state behavior discussed above. In the purely graph-based regime, the walk is governed entirely by the lower channel and transport remains relatively efficient. In the fully simplicial regime, the upper channel becomes uniformly available and the dynamics recover a more homogeneous character. The largest first-passage times arise between these two limits, where face dilution introduces substantial heterogeneity in the upper connectivity without completely suppressing the contribution of triangles. Thus, the competition between lower transport and partial higher-order connectivity produces a regime of maximal dynamical slowdown.

Figure~\ref{MFPT}(b) provides a global view of the transport behavior by mapping $\langle\mathrm{MFPT}\rangle$ over the $(p,q)$ parameter space. The transport bottleneck is confined to a well-defined intermediate region, where partial face dilution coexists with small to moderate values of $q$. In agreement with Fig.~\ref{MFPT}(a), the largest values of $\langle\mathrm{MFPT}\rangle$ occur precisely in this region, while transport remains comparatively efficient in both the graph-dominated ($q\rightarrow1$) and highly connected simplicial regimes.

While the mean first-passage time identifies the transport bottleneck, it does not reveal how transport is distributed across individual trajectories. To gain further insight into the origin of this slowdown, we examine the distribution of first-passage times between a fixed pair of edges.

For the fixed source-target pair considered below, we characterize the distribution through its empirical survival function,
\begin{equation}
S(\tau)=\Pr\!\left(T_{e\to e'}>\tau\right),
\end{equation}
which gives the probability that the target edge has not yet been reached after $\tau$ steps.

Figure~\ref{Fig08} shows $S(\tau)$ for the fixed source and target edges $e_1=(1,2)$ and $e_{120}=(15,16)$, respectively, at $n=16$ and $q=0.1$. We consider three representative face densities: a dilute regime before the transport bottleneck, $p=0.005$; a value close to the maximum of $\langle\mathrm{MFPT}\rangle$, $p=0.06$; and a denser simplicial regime after the maximum, $p=0.5$.

\begin{figure}
\centering
\includegraphics[trim={0.52cm 12cm 4cm 0},clip, width=1.1\linewidth]{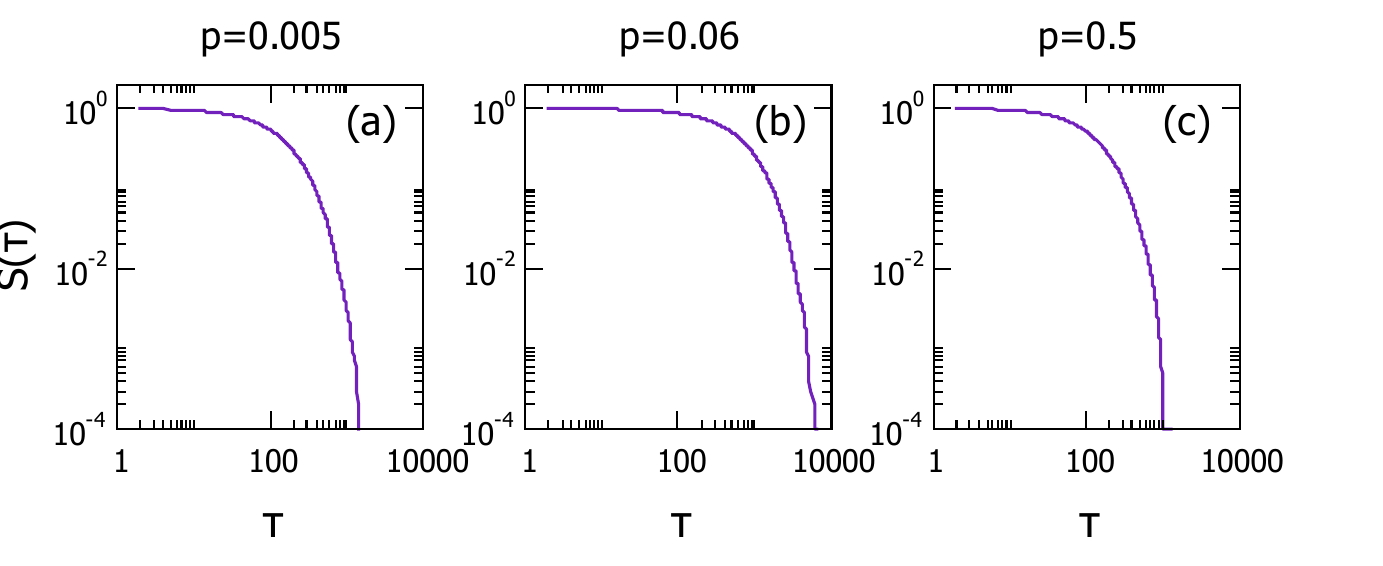}

\caption{
Empirical survival probability of the first-passage time between the fixed
source edge $e_1=(1,2)$ and target edge $e_{120}=(15,16)$ for $n=16$ and
$q=0.1$. Panels (a)--(c) correspond to $p=0.005$, $0.06$, and $0.5$,
respectively, representing regimes before, near, and after the maximum of
$\langle\mathrm{MFPT}\rangle$.
}
\label{Fig08}
\end{figure}
The survival curves reveal a dynamical feature that is not conveyed by the
mean alone. Before and after the bottleneck, the survival probability decays
rapidly, indicating that long first-passage trajectories are comparatively
rare. Near the maximum of $\langle\mathrm{MFPT}\rangle$, however, the decay
extends over a substantially longer time scale. The transport bottleneck identified in Fig.~\ref{MFPT}
is therefore accompanied by a broadening of the first-passage-time distribution and by an enhanced probability of unusually
long trajectories.

Taken together, the structural, spectral, stationary-state, and first-passage analyses consistently identify the same intermediate regime generated by partial face retention. The heterogeneous upper connectivity simultaneously slows relaxation, localizes the stationary state, increases the mean first-passage time, and broadens the distribution of first-passage trajectories. These independent observations demonstrate that the resulting transport bottleneck is a robust dynamical consequence of partial higher-order connectivity rather than a feature of any particular observable.

\subsection{The singular limit $q=0$}

The results above show that partial face dilution produces a pronounced transport bottleneck whenever the upper channel contributes to the dynamics. To understand the mechanism behind this behavior, we consider the limiting case $q=0$, where the effective transition operator is given by
\begin{equation}
P_{\mathrm{eff}}
=
P^{\downarrow}(I-A)
+
P^{\uparrow}A,
\label{eq:Peff_q0}
\end{equation}
where the diagonal matrix $A$ has entries equal to one for active edges, namely those incident to at least one face, and zero otherwise.
Unlike previous cases, Eq.~(\ref{eq:Peff_q0}) does not define a single transition rule over the entire edge space. Active edges evolve through the upper transition operator $P^{\uparrow}$, while inactive edges follow $P^{\downarrow}$. Consequently, the transition mechanism becomes state dependent, making the dynamics qualitatively different from those for $q>0$.
To characterize the resulting dynamics, we analyze the spectrum of
$P_{\mathrm{eff}}$. We denote by
\begin{equation*}
m_1=\mathrm{mult}(\lambda=1)
\end{equation*}
the algebraic multiplicity of the unit eigenvalue of
$P_{\mathrm{eff}}$. For an irreducible stochastic matrix, one has $m_1=1$, whereas $m_1>1$ indicates the existence of multiple closed communicating classes and, therefore, the loss of ergodicity.
For each value of $p$, we compute $m_1$ for every realization and evaluate its ensemble average, $\langle m_1\rangle$.

\begin{figure}
\centering
\includegraphics[trim={0.3cm 12.2cm 9.1cm 0.5cm},clip, width=0.85\linewidth]{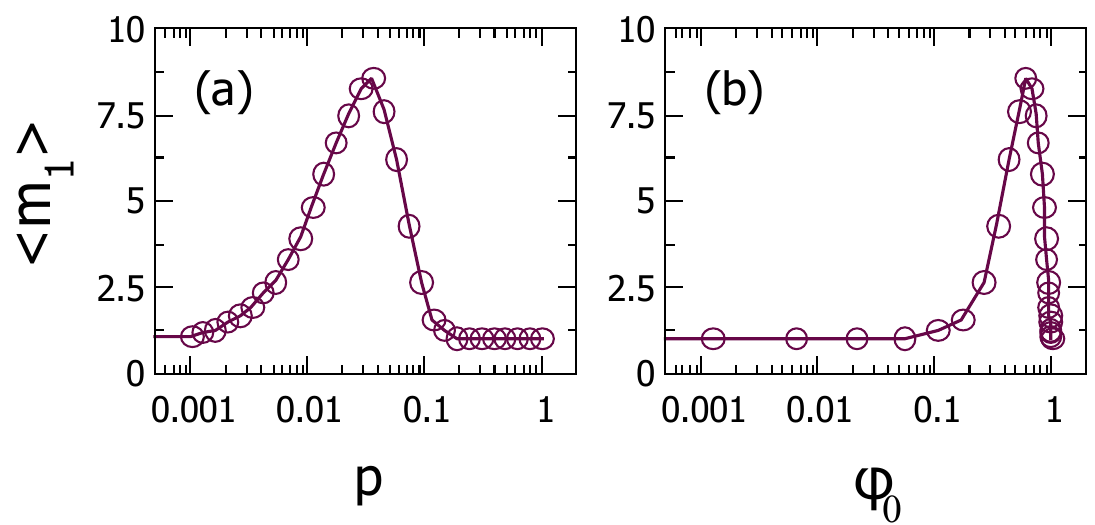}
\caption{ Average unit eigenvalue multiplicity $\langle m_1\rangle$ in the singular limit $q=0$. (a) As a function of the face-retention probability $p$. (b) As a function of the inactive-edge fraction $\langle\phi_0\rangle$. }
\label{Fig:q0}
\end{figure}
Figure~\ref{Fig:q0}(a) shows that the singular nature of the dynamics is confined to an intermediate interval of face-retention probabilities. At $p=0$, one has $A=0$, so that $P_{\mathrm{eff}}=P^\downarrow$ and the unit eigenvalue is nondegenerate. This behavior persists for sufficiently small values of $p$, where only a few active edges are present. As $p$ increases, however, active and inactive edges coexist throughout the complex, producing a rapid increase in the multiplicity of the unit eigenvalue. The average multiplicity reaches a maximum at intermediate face densities and then gradually decreases as the upper channel becomes available almost everywhere.

The structural origin of this behavior becomes clearer when the same quantity is represented as a function of the inactive-edge fraction $\langle\phi_0\rangle$, shown in Fig.~\ref{Fig:q0}(b). Expressing the results in terms of $\langle\phi_0\rangle$ provides a structural interpretation of the spectral degeneracy in terms of the availability of the upper transition channel.

The largest values of $\langle m_1\rangle$ occur at intermediate values of $\langle\phi_0\rangle$, where active and inactive edges coexist in comparable proportions. This observation provides a direct structural interpretation of the singular dynamics at $q=0$. When nearly all edges are either active or inactive, a single transition mechanism dominates throughout the complex and the walk remains globally coherent. In contrast, intermediate values of $\langle\phi_0\rangle$ maximize the spatial coexistence of the two local transport rules, fragmenting the dynamics into multiple recurrent classes and producing the observed spectral degeneracy.

This singular behavior helps elucidate the transport bottleneck observed for $q>0$. Although the lower transition remains accessible for $q>0$, the dynamics still inherit the strong heterogeneity generated by the intermittent availability of the upper channel. The limiting case therefore suggests that the slowdown observed throughout the previous sections is closely associated with the heterogeneous availability of higher-order transitions rather than with face dilution alone.

\section{Conclusions and discussion}
\label{conclusions}

In this work, we introduced an incidence-based random walk on the edge set of a random two-dimensional simplicial complex with a complete $1$-skeleton and independently retained triangular faces. The model combines two transport channels, one mediated by vertices and the other by triangles, and remains well defined even when some edges are not incident to any face. This construction isolates the role of higher-order connectivity on transport without modifying the underlying pairwise support of the walk.

Our results show that the dynamical effects of higher-order interactions are strongest in an intermediate regime of face retention. Partial face dilution generates a heterogeneous distribution of upper connectivity, characterized by broad fluctuations in the triangle degree and the coexistence of active and inactive edges. This structural heterogeneity conspicuously impacts 
the dynamics: the second-largest eigenvalue modulus $\langle\lambda_*\rangle$, the stationary inverse participation ratio $\langle\mathrm{IPR}(\pi)\rangle$, and the mean first-passage time $\langle\mathrm{MFPT}\rangle$ all exhibit pronounced non-monotonic behavior, attaining  their largest values. indicating poor transport, at intermediate face densities. The survival probability of first-passage times further shows that this regime is accompanied by an enhanced probability of exceptionally long trajectories, demonstrating that the transport bottleneck is not only an average effect but also a property of the full first-passage-time distribution.

The singular limit $q=0$ provides additional insight into this behavior. In this limit, the local transport mechanism becomes state dependent because active and inactive edges evolve through different transition operators. The coexistence of both transport rules causes a marked increase in the multiplicity of the unit eigenvalue, revealing multiple recurrent classes when the upper channel is only partially available. Although this singular behavior disappears for any $q>0$, it highlights the dynamical consequences of heterogeneous higher-order connectivity and offers a useful perspective for interpreting the transport slowdown observed in the intermediate regime.

The present framework finds applications in systems where pairwise interactions are not fixed but evolve in time. In our case, the stochastic dynamics of edge activation are modeled by a random walk on the edge space, whose transition probabilities incorporate the underlying higher-order structure of simplicial complexes. It can also be extended in several directions. Natural generalizations include random simplicial complexes with non-complete $1$-skeletons, higher-dimensional simplicial complexes, and continuous-time versions of the current dynamics. It would also be interesting to establish connections with diffusion processes generated by Hodge Laplacians and other higher-order operators. More generally, the proposed incidence-based framework provides a simple setting to investigate how heterogeneous higher-order interactions shape spectral and transport properties beyond pairwise network dynamics.

\begin{acknowledgments}
C.T.M.-M. thanks support from CONAHCYT (CVU No.~784756).
This work was supported by UNAM-PAPIIT IN110626.
\end{acknowledgments}

\appendix

\section{Proofs of basic properties}

\subsection{Proof that $P_{\mathrm{eff}}$ is column-stochastic}

\begin{proof}
Since each edge has exactly two endpoints, the columns of $P_{10}$ sum to one. Since $P_{01}$ chooses uniformly among the edges incident to a vertex, its columns also sum to one. Hence
\begin{equation}
\mathbf{1}^{\top} P^{\downarrow}
=
\mathbf{1}^{\top} P_{01}P_{10}
=
\mathbf{1}^{\top}.
\end{equation}

Likewise, the columns of $P_{21}$ sum to one because every triangle has exactly three edges. The columns of $P_{12}$ sum to $a_e$, namely
\begin{equation}
\mathbf{1}^{\top} P_{12}
=
\mathbf{1}^{\top} A,
\end{equation}
because the sum is $1$ when $d_T(e)>0$ and $0$ otherwise. Therefore,
\begin{equation}
\mathbf{1}^{\top} P^{\uparrow}
=
\mathbf{1}^{\top} P_{21}P_{12}
=
\mathbf{1}^{\top} A.
\end{equation}
Using \eqref{eq:Peff},
\begin{equation}
\mathbf{1}^{\top} P_{\mathrm{eff}}(q)
=
\mathbf{1}^{\top}\!\bigl[I-(1-q)A\bigr]
+
(1-q)\mathbf{1}^{\top} A
=
\mathbf{1}^{\top}.
\end{equation}
Hence $P_{\mathrm{eff}}(q)$ is column-stochastic.
\end{proof}

\section{Basic properties of the lower, upper, and effective walks}
\label{sec:basic-properties}

We collect here the main probabilistic properties of the lower, upper, and effective kernels.

Define the set of active edges by
\begin{equation}
X_1^+ = \{e \in X_1 : d_T(e)>0\}.
\end{equation}

\subsection{Lower walk}

Recall that
\begin{equation}
P^{\downarrow}
=
P_{01} P_{10}
=
\frac{1}{2}\,\mathbf{B}_1^{\top} D_V^{-1}\mathbf{B}_1,
\end{equation}
with entries
\begin{equation}
P^{\downarrow}_{e'e}
=
\frac{1}{2}\sum_{v \in X_0}
\frac{\mathbf{1}_{\{v \subset e\}} \mathbf{1}_{\{v \subset e'\}}}{d_V(v)}.
\end{equation}

\begin{prop}
The lower walk $P^{\downarrow}$ is symmetric and column-stochastic. Consequently, it is also row-stochastic, hence doubly stochastic.
\end{prop}

\begin{proof}
Since $D_V^{-1}$ is diagonal,
\begin{equation}
\bigl(P^{\downarrow}\bigr)^{\top}
=
\frac{1}{2}\bigl(\mathbf{B}_1^{\top} D_V^{-1}\mathbf{B}_1\bigr)^{\top}
=
\frac{1}{2}\,\mathbf{B}_1^{\top} D_V^{-1}\mathbf{B}_1
=
P^{\downarrow}.
\end{equation}
Thus $P^{\downarrow}$ is symmetric.

Fix $e \in X_1$. Then
\begin{equation}
\begin{aligned}
\sum_{e' \in X_1} P^{\downarrow}_{e'e}
&=
\sum_{e' \in X_1}
\frac{1}{2}\sum_{v \in X_0}
\frac{\mathbf{1}_{\{v \subset e\}} \mathbf{1}_{\{v \subset e'\}}}{d_V(v)}
\\
&=
\frac{1}{2}\sum_{v \subset e}\frac{1}{d_V(v)}
\sum_{e' \ni v} 1.
\end{aligned}
\end{equation}
Since the number of edges incident to $v$ is exactly $d_V(v)$,
\begin{equation}
\sum_{e' \in X_1} P^{\downarrow}_{e'e}
=
\frac{1}{2}\sum_{v \subset e} 1
=
1,
\end{equation}
because every edge has exactly two endpoints. Therefore each column sums to one. Since $P^{\downarrow}$ is symmetric, it is also row-stochastic.
\end{proof}

\begin{cor}
The uniform distribution on edges,
\begin{equation}
\pi^{\downarrow}_e = \frac{1}{|X_1|},
\qquad e \in X_1,
\end{equation}
is stationary for $P^{\downarrow}$. Moreover, $P^{\downarrow}$ is reversible with respect to $\pi^{\downarrow}$.
\end{cor}

\begin{proof}
Since $P^{\downarrow}$ is doubly stochastic,
\begin{equation}
P^{\downarrow}\pi^{\downarrow} = \pi^{\downarrow}.
\end{equation}
By symmetry,
\begin{equation}
\pi^{\downarrow}_e P^{\downarrow}_{e'e}
=
\frac{1}{|X_1|} P^{\downarrow}_{e'e}
=
\frac{1}{|X_1|} P^{\downarrow}_{ee'}
=
\pi^{\downarrow}_{e'} P^{\downarrow}_{ee'},
\end{equation}
which is the detailed balance condition.
\end{proof}

For $e=\{u,v\}$,
\begin{equation}
P^{\downarrow}_{ee}
=
\frac{1}{2}\left(\frac{1}{d_V(u)}+\frac{1}{d_V(v)}\right),
\end{equation}
so the lower walk always has a positive self-loop.

\subsection{Upper walk}

Recall that
\begin{equation}
P^{\uparrow}
=
P_{21} P_{12}
=
\frac{1}{3}\,\mathbf{B}_2 \mathbf{B}_2^{\top} \widetilde{D}_T^{-1},
\end{equation}
and for active edges $e \in X_1^+$,
\begin{equation}
P^{\uparrow}_{e'e}
=
\frac{1}{3}\sum_{f \in X_2}
\frac{\mathbf{1}_{\{e \subset f\}} \mathbf{1}_{\{e' \subset f\}}}{d_T(e)}.
\end{equation}
If $d_T(e)=0$, the whole column indexed by $e$ vanishes, so $P^{\uparrow}$ is not a Markov kernel on all of $X_1$, but it is one on $X_1^+$.

\begin{prop}
The restriction of $P^{\uparrow}$ to $X_1^+$ is column-stochastic.
\end{prop}

\begin{proof}
Fix $e \in X_1^+$, so that $d_T(e)>0$. Then
\begin{equation}
\begin{aligned}
\sum_{e' \in X_1} P^{\uparrow}_{e'e}
&=
\sum_{e' \in X_1} \frac{1}{3}\sum_{f \in X_2}
\frac{\mathbf{1}_{\{e \subset f\}} \mathbf{1}_{\{e' \subset f\}}}{d_T(e)}
\\
&=
\frac{1}{3d_T(e)}
\sum_{f \supset e} \sum_{e' \subset f} 1.
\end{aligned}
\end{equation}
Each triangle has exactly three boundary edges, hence
\begin{equation}
\sum_{e' \in X_1} P^{\uparrow}_{e'e}
=
\frac{1}{3d_T(e)} \sum_{f \supset e} 3
=
1.
\end{equation}
\end{proof}

\begin{prop}
On $X_1^+$, the upper walk is reversible with respect to
\begin{equation}
\pi^{\uparrow}_e
=
\frac{d_T(e)}{\sum_{\tilde e \in X_1^+} d_T(\tilde e)}.
\end{equation}
In particular, $\pi^{\uparrow}$ is stationary for the restricted upper walk.
\end{prop}

\begin{proof}
For $e,e' \in X_1^+$,
\begin{equation}
\begin{aligned}
\pi^{\uparrow}_e P^{\uparrow}_{e'e}
&=
\frac{d_T(e)}{Z}
\cdot
\frac{1}{3}\sum_{f \in X_2}
\frac{\mathbf{1}_{\{e \subset f\}} \mathbf{1}_{\{e' \subset f\}}}{d_T(e)}
\\
&=
\frac{1}{3Z}\sum_{f \in X_2}
\mathbf{1}_{\{e \subset f\}} \mathbf{1}_{\{e' \subset f\}},
\end{aligned}
\end{equation}
where
\begin{equation}
Z := \sum_{\tilde e \in X_1^+} d_T(\tilde e).
\end{equation}
The right-hand side is symmetric in $e$ and $e'$, so detailed balance holds.
\end{proof}

Moreover, for every active edge $e \in X_1^+$,
\begin{equation}
P^{\uparrow}_{ee} = \frac{1}{3}.
\end{equation}

\subsection{Effective walk}

The effective transition is
\begin{equation}
P_{\mathrm{eff}}(q)
=
P^{\downarrow}\!\bigl[I-(1-q)A\bigr]
+
(1-q)\,P^{\uparrow}A,
\qquad q \in [0,1].
\end{equation}

\begin{prop}
For every $q \in [0,1]$, the effective walk admits at least one stationary distribution, namely a probability vector $\pi$ satisfying
\begin{equation}
P_{\mathrm{eff}}(q)\,\pi = \pi.
\end{equation}
\end{prop}

\begin{proof}
Since $P_{\mathrm{eff}}(q)$ is a finite column-stochastic matrix, it defines a Markov chain on the finite state space $X_1$. Standard finite-state Markov-chain theory then guarantees the existence of at least one stationary distribution.
\end{proof}

If $q>0$, every allowed lower transition remains present in the effective kernel.

\begin{prop}
Assume $q>0$. Then
\begin{equation}
\bigl(P_{\mathrm{eff}}\bigr)_{e'e} > 0
\qquad \text{whenever} \qquad
P^{\downarrow}_{e'e} > 0.
\end{equation}
Hence the support of $P_{\mathrm{eff}}(q)$ contains the support of the lower walk. In particular, if the line graph of the $1$-skeleton is connected, then $P_{\mathrm{eff}}(q)$ is irreducible.
\end{prop}

\begin{proof}
If $d_T(e)=0$, then
\begin{equation}
P_{\mathrm{eff}}(\cdot,e) = P^{\downarrow}(\cdot,e),
\end{equation}
and the claim is immediate.

If $d_T(e)>0$, then
\begin{equation}
P_{\mathrm{eff}}(\cdot,e)
=
q\,P^{\downarrow}(\cdot,e) + (1-q)\,P^{\uparrow}(\cdot,e).
\end{equation}
Since $q>0$ and $P^{\downarrow}(\cdot,e)\ge 0$ entrywise,
\begin{equation}
P^{\downarrow}_{e'e} > 0
\quad \Longrightarrow \quad
\bigl(P_{\mathrm{eff}}\bigr)_{e'e} \ge q\,P^{\downarrow}_{e'e} > 0.
\end{equation}
Thus every transition allowed by the lower walk is also allowed by the effective walk.
\end{proof}

\begin{cor}
Assume $q>0$ and that the line graph of the $1$-skeleton is connected. Then the effective walk is irreducible and aperiodic. Consequently, it admits a unique stationary distribution $\pi^{\mathrm{eff}}$, and for any initial condition $p(0)$ one has
\begin{equation}
\lim_{t\to\infty} p(t) = \pi^{\mathrm{eff}}.
\end{equation}
\end{cor}

\begin{proof}
Irreducibility follows from the previous proposition. Aperiodicity follows from the fact that
\begin{equation}
\bigl(P_{\mathrm{eff}}\bigr)_{ee} > 0
\qquad \forall e \in X_1,
\end{equation}
as established in Section~\ref{sec:model}. The convergence statement is then standard.
\end{proof}

\subsection{Remarks on symmetry and reversibility}

The lower walk is symmetric, hence reversible with respect to the uniform distribution. The upper walk, restricted to active edges, is reversible with respect to $\pi^{\uparrow}$. In general, these two invariant measures do not coincide, so their combination in $P_{\mathrm{eff}}$ need not be reversible.

Indeed, for active edges $e,e' \in X_1^+$ one has
\begin{equation}
P^{\uparrow}_{e'e}
=
\frac{N_2(e,e')}{3\,d_T(e)},
\end{equation}
where
\begin{equation}
N_2(e,e')
=
\#\{f \in X_2 : e \subset f,\ e' \subset f\}.
\end{equation}
Although $N_2(e,e') = N_2(e',e)$ is symmetric, the normalization by $d_T(e)$ makes $P^{\uparrow}$ generally non-symmetric unless the upper degrees are suitably balanced. Consequently, whenever $q<1$ and the upper degrees are non-uniform, $P_{\mathrm{eff}}$ is typically non-symmetric.

In highly regular situations, however, the effective walk may recover symmetry. This happens, for instance, in the complete $2$-dimensional simplicial complex on $n$ vertices, where all vertices have the same degree and all edges belong to the same number of triangles. In that case both transport channels are homogeneous, and the effective kernel becomes symmetric and therefore doubly stochastic.

\bibliographystyle{apsrev4-2}
\bibliography{references}

@book{brouwer2011spectra,
  title={Spectra of graphs},
  author={Brouwer, Andries E and Haemers, Willem H},
  year={2011},
  publisher={Springer Science \& Business Media}
}

@book{Cvetkovic2004,
  author    = {Drago{\v{s}} Cvetkovi{\'c} and Peter Rowlinson and Slobodan Simi{\'c}},
  title     = {Spectral Generalizations of Line Graphs: On Graphs with Least Eigenvalue -2},
  publisher = {Cambridge University Press},
  year      = {2004}
}

@article{Masuda2017,
  title={Random walks and diffusion on networks},
  author={Masuda, Naoki and Porter, Mason A and Lambiotte, Renaud},
  journal={Physics reports},
  volume={716},
  pages={1--58},
  year={2017},
  publisher={Elsevier}
}

@article{Schaub2020,
  title={Random walks on simplicial complexes and the normalized Hodge 1-Laplacian},
  author={Schaub, Michael T and Benson, Austin R and Horn, Paul and Lippner, Gabor and Jadbabaie, Ali},
  journal={SIAM Review},
  volume={62},
  number={2},
  pages={353--391},
  year={2020},
  publisher={SIAM}
}

@misc{FebbeArXiv2026b,
      title={Random Walks Across Dimensions: Exploring Simplicial Complexes}, 
      author={Diego Febbe and Duccio Fanelli and Timoteo Carletti},
      year={2026},
      eprint={2601.16086},
      archivePrefix={arXiv},
      primaryClass={cond-mat.stat-mech},
      url={https://arxiv.org/abs/2601.16086}, 
}

@article{MajhiJRoySocInterface2022,
    author = {Majhi, Soumen and Perc, Matjaž and Ghosh, Dibakar},
    title = {Dynamics on higher-order networks: a review},
    journal = {Journal of The Royal Society Interface},
    volume = {19},
    number = {188},
    pages = {20220043},
    year = {2022},
    month = {03},
    issn = {1742-5689},
    doi = {10.1098/rsif.2022.0043},
    url = {https://doi.org/10.1098/rsif.2022.0043},
    eprint = {https://royalsocietypublishing.org/rsif/article-pdf/doi/10.1098/rsif.2022.0043/926513/rsif.2022.0043.pdf},
}

@article{LambiotteNatPhys2019,
author={Lambiotte, Renaud
and Rosvall, Martin
and Scholtes, Ingo},
title={From networks to optimal higher-order models of complex systems},
journal={Nature Physics},
year={2019},
month={Apr},
day={01},
volume={15},
number={4},
pages={313-320},
issn={1745-2481},
doi={10.1038/s41567-019-0459-y},
url={https://doi.org/10.1038/s41567-019-0459-y}
}

@book{BianconiBook2021,
  author    = {G. Bianconi},
  title     = {Higher-Order Networks: An Introduction to Simplicial Complexes},
  publisher = {Cambridge University Press},
  address   = {Cambridge},
  year      = {2021},
  doi       = {10.1017/9781108770996}
}

@article{BensonScience2016,
author = {Austin R. Benson  and David F. Gleich  and Jure Leskovec },
title = {Higher-order organization of complex networks},
journal = {Science},
volume = {353},
number = {6295},
pages = {163-166},
year = {2016},
doi = {10.1126/science.aad9029},
URL = {https://www.science.org/doi/abs/10.1126/science.aad9029},
eprint = {https://www.science.org/doi/pdf/10.1126/science.aad9029}}

@article{BattistonNatPhys2021,
author={Battiston, Federico
and Amico, Enrico
and Barrat, Alain
and Bianconi, Ginestra
and Ferraz de Arruda, Guilherme
and Franceschiello, Benedetta
and Iacopini, Iacopo
and K{\'e}fi, Sonia
and Latora, Vito
and Moreno, Yamir
and Murray, Micah M.
and Peixoto, Tiago P.
and Vaccarino, Francesco
and Petri, Giovanni},
title={The physics of higher-order interactions in complex systems},
journal={Nature Physics},
year={2021},
month={Oct},
day={01},
volume={17},
number={10},
pages={1093-1098},
issn={1745-2481},
doi={10.1038/s41567-021-01371-4},
url={https://doi.org/10.1038/s41567-021-01371-4}
}

@article{MillanNatPhys2025,
author={Mill{\'a}n, Ana P.
and Sun, Hanlin
and Giambagli, Lorenzo
and Muolo, Riccardo
and Carletti, Timoteo
and Torres, Joaqu{\'i}n J.
and Radicchi, Filippo
and Kurths, J{\"u}rgen
and Bianconi, Ginestra},
title={Topology shapes dynamics of higher-order networks},
journal={Nature Physics},
year={2025},
month={Mar},
day={01},
volume={21},
number={3},
pages={353-361},
issn={1745-2481},
doi={10.1038/s41567-024-02757-w},
url={https://doi.org/10.1038/s41567-024-02757-w}
}

@article{SawadaPhysRevLett2024,
  title = {Role of Topology in Relaxation of One-Dimensional Stochastic Processes},
  author = {Sawada, Taro and Sone, Kazuki and Hamazaki, Ryusuke and Ashida, Yuto and Sagawa, Takahiro},
  journal = {Phys. Rev. Lett.},
  volume = {132},
  issue = {4},
  pages = {046602},
  numpages = {6},
  year = {2024},
  month = {Jan},
  publisher = {American Physical Society},
  doi = {10.1103/PhysRevLett.132.046602},
  url = {https://link.aps.org/doi/10.1103/PhysRevLett.132.046602}
}

@article{RiascosJCompNetworks2021,
    author = {Riascos, A P and Mateos, José L},
    title = {Random walks on weighted networks: a survey of local and non-local dynamics},
    journal = {Journal of Complex Networks},
    volume = {9},
    number = {5},
    pages = {cnab032},
    year = {2021},
    month = {10},
    issn = {2051-1329},
    doi = {10.1093/comnet/cnab032},
    url = {https://doi.org/10.1093/comnet/cnab032},
    eprint = {https://academic.oup.com/comnet/article-pdf/9/5/cnab032/40545326/cnab032.pdf},
}

@article{CarlettiPhysRevE2020,
  title = {Random walks on hypergraphs},
  author = {Carletti, Timoteo and Battiston, Federico and Cencetti, Giulia and Fanelli, Duccio},
  journal = {Phys. Rev. E},
  volume = {101},
  issue = {2},
  pages = {022308},
  numpages = {19},
  year = {2020},
  month = {Feb},
  publisher = {American Physical Society},
  doi = {10.1103/PhysRevE.101.022308},
  url = {https://link.aps.org/doi/10.1103/PhysRevE.101.022308}
}

@article{TraversaPhysRevE2023,
  title = {From unbiased to maximal-entropy random walks on hypergraphs},
  author = {Traversa, Pietro and de Arruda, Guilherme Ferraz and Moreno, Yamir},
  journal = {Phys. Rev. E},
  volume = {109},
  issue = {5},
  pages = {054309},
  numpages = {15},
  year = {2024},
  month = {May},
  publisher = {American Physical Society},
  doi = {10.1103/PhysRevE.109.054309},
  url = {https://link.aps.org/doi/10.1103/PhysRevE.109.054309}
}

@misc{FebbeArXiv2026a,
      title={Optimal Navigation on Simplicial Complexes}, 
      author={Diego Febbe and Duccio Fanelli and Gianluca Peri and Timoteo Carletti},
      year={2026},
      eprint={2607.29450},
      archivePrefix={arXiv},
      primaryClass={cond-mat.stat-mech},
      url={https://arxiv.org/abs/2607.29450}, 
}

@article{BensonSIAM2017,
  author  = {Benson, Austin R. and Gleich, David F. and Lim, Lek-Heng},
  title   = {The Spacey Random Walk: A Stochastic Process for Higher-Order Data},
  journal = {SIAM Review},
  volume  = {59},
  number  = {2},
  pages   = {321--345},
  year    = {2017},
  doi     = {10.1137/16M1074023}
}

@inproceedings{HayashiBook2020,
  author    = {Hayashi, Koby and Aksoy, Sinan G. and Park, Cheong Hee and Park, Haesun},
  title     = {Hypergraph Random Walks, {L}aplacians, and Clustering},
  booktitle = {Proceedings of the 29th ACM International Conference on Information \& Knowledge Management (CIKM)},
  pages     = {495--504},
  year      = {2020},
  doi       = {10.1145/3340531.3412034}
}

@article{WangPhyRep2024,
  author  = {Wang, Wei and Nie, Yanyi and Li, Wenyao and Lin, Tao and Shang, Ming-Sheng and Su, Song and Tang, Yong and Zhang, Yi-Cheng and Sun, Gui-Quan},
  title   = {Epidemic spreading on higher-order networks},
  journal = {Physics Reports},
  volume  = {1056},
  pages   = {1--70},
  year    = {2024},
  doi     = {10.1016/j.physrep.2024.01.003}
}

@article{IacopiniNature2019,
  author  = {Iacopini, Iacopo and Petri, Giovanni and Barrat, Alain and Latora, Vito},
  title   = {Simplicial models of social contagion},
  journal = {Nature Communications},
  volume  = {10},
  pages   = {2485},
  year    = {2019},
  doi     = {10.1038/s41467-019-10431-6}
}

@article{ErdosRenyi1959,
  author  = {Erd{\H{o}}s, Paul and R{\'e}nyi, Alfr{\'e}d},
  title   = {On random graphs {I}},
  journal = {Publicationes Mathematicae Debrecen},
  volume  = {6},
  pages   = {290--297},
  year    = {1959}
}

@incollection{Kahle2014Survey,
  author    = {Kahle, Matthew},
  title     = {Topology of random simplicial complexes: a survey},
  booktitle = {Algebraic Topology: Applications and New Directions},
  series    = {Contemporary Mathematics},
  volume    = {620},
  pages     = {201--221},
  publisher = {American Mathematical Society},
  year      = {2014},
  doi       = {10.1090/conm/620/12365}
}

@article{LinialMeshulam2006,
  author  = {Linial, Nathan and Meshulam, Roy},
  title   = {Homological connectivity of random 2-complexes},
  journal = {Combinatorica},
  volume  = {26},
  number  = {4},
  pages   = {475--487},
  year    = {2006},
  doi     = {10.1007/s00493-006-0027-9}
}

@article{Wegner1980,
  author  = {Wegner, F.},
  title   = {Inverse participation ratio in $2+\epsilon$ dimensions},
  journal = {Zeitschrift f{\"u}r Physik B},
  volume  = {36},
  pages   = {209--214},
  year    = {1980},
  doi     = {10.1007/BF01325284}
}

@article{MartinPRE2014,
  author  = {Martin, Travis and Zhang, Xiao and Newman, M. E. J.},
  title   = {Localization and centrality in networks},
  journal = {Physical Review E},
  volume  = {90},
  pages   = {052808},
  year    = {2014},
  doi     = {10.1103/PhysRevE.90.052808}
}

@article{MartinezMartinez2019Entropy,
  author  = {Mart{\'i}nez-Mart{\'i}nez, C. T. and M{\'e}ndez-Berm{\'u}dez, J. A.},
  title   = {Information Entropy of Tight-Binding Random Networks with Losses and Gain: Scaling and Universality},
  journal = {Entropy},
  volume  = {21},
  number  = {1},
  pages   = {86},
  year    = {2019},
  doi     = {10.3390/e21010086}
}

@article{MendezBermudez2015,
  author  = {M{\'e}ndez-Berm{\'u}dez, J. A. and Alcazar-L{\'o}pez, A. and Mart{\'i}nez-Mendoza, A. J. and Rodrigues, Francisco A. and Peron, Thomas K. DM.},
  title   = {Universality in the spectral and eigenfunction properties of random networks},
  journal = {Physical Review E},
  volume  = {91},
  pages   = {032122},
  year    = {2015},
  doi     = {10.1103/PhysRevE.91.032122}
}

@article{Condamin2007Nature,
  author  = {Condamin, S. and B{\'e}nichou, O. and Tejedor, V. and Voituriez, R. and Klafter, J.},
  title   = {First-passage times in complex scale-invariant media},
  journal = {Nature},
  volume  = {450},
  number  = {7166},
  pages   = {77--80},
  year    = {2007},
  doi     = {10.1038/nature06201}
}

@article{NohRieger2004,
  author  = {Noh, Jae Dong and Rieger, Heiko},
  title   = {Random Walks on Complex Networks},
  journal = {Physical Review Letters},
  volume  = {92},
  number  = {11},
  pages   = {118701},
  year    = {2004},
  doi     = {10.1103/PhysRevLett.92.118701}
}

@article{Tejedor2009GMFPT,
  author  = {Tejedor, V. and B{\'e}nichou, O. and Voituriez, R.},
  title   = {Global mean first-passage times of random walks on complex networks},
  journal = {Physical Review E},
  volume  = {80},
  pages   = {065104(R)},
  year    = {2009},
  doi     = {10.1103/PhysRevE.80.065104}
}

@book{Redner2001,
  author    = {Redner, Sidney},
  title     = {A Guide to First-Passage Processes},
  publisher = {Cambridge University Press},
  address   = {Cambridge, UK},
  year      = {2001}
}
\end{document}